\documentclass[aps,pra,showpacs,twoside,twocolumn,10pt]{revtex4-2}
\usepackage[colorlinks=true, citecolor=blue, urlcolor=blue, linkcolor = blue ]{hyperref}

\usepackage{amsthm,graphicx,bm,amsmath,xcolor,braket,plain,color,amsthm,amsfonts,tabularx,graphicx,bbm,mathtools,esvect,wrapfig,verbatim,enumitem,dcolumn,amssymb,appendix,physics,fmtcount,booktabs,csquotes,epsfig,times,geometry,array,mathrsfs,graphics,epstopdf,latexsym,comment,cancel}

\usepackage[normalem]{ulem}

\usepackage[mathscr]{euscript}
\def\Tr{\operatorname{Tr}}

\newtheorem{theorem}{Theorem}

\newtheorem{remark}{Remark}

\newtheorem{corollary}{Corollary}

\begin{document}

\title{Optimal quantum precision in noise estimation: Is entanglement necessary?}

\author{Shuva Mondal}
\email{mondalshuva01@gmail.com}

\author{Priya Ghosh}
\email{priyaghosh1155@gmail.com}

\author{Ujjwal Sen}
\email{ujjwalsen0601@gmail.com}

\affiliation{Harish-Chandra Research Institute, A CI of Homi Bhabha National Institute, Chhatnag Road, Jhunsi, Prayagraj 211 019, India}

\begin{abstract}
We ask whether the optimal probe is entangled, and if so,  what is its character and amount, for estimating 
the noise parameter of a large class of local quantum encoding processes that we refer to  as vector encoding, examples of which include the local depolarizing and bit-flip channels. 
We first establish that vector encoding is invariably ``continuously commutative'' for optimal probes. 
We utilize this result to deal with the queries about entanglement in the optimal probe. 
We show that for estimating noise extent of the two-party arbitrary-dimensional local depolarizing channel, there is a descending staircase of optimal-probe entanglement for increasing depolarizing strength. 
For the multi-qubit case, the analysis again leads to a staircase, but which can now be monotonic or 
not, depending on the multiparty entanglement measure used. 
We also find that when sufficiently high depolarizing noise is to be estimated, fully product multiparty states are the only choice for being optimal probes. In many cases, for even moderately high depolarizing noise, fully product states are optimal. For two-qubit local bit-flip channels, the continuous commutativity of the channel and optimal probe implies that a product state suffices for obtaining the optimal precision.
\end{abstract}

\maketitle

\section{Introduction}
Quantum parameter estimation provides a tool to estimate a parameter in a quantum setup as precisely as possible within the framework of quantum mechanics. Recently, quantum metrology~\cite{RevModPhys.90.035005,PhysRevLett.96.010401,Giovannetti2004,RevModPhys.89.035002,RevModPhys.90.035006,Pirandola2018,Helstrom1969,Holevo1973,Braunstein_PRL_1994,Holevo2011} has gained much attention, in particular for improving the precision using quantum resources like entanglement~\cite{Cirac96,ent-optimal,Riedel2010,Giovannetti2011,bhattacharyya2024even,group_debarupa_2025}, quantum coherence~\cite{Giovannetti2004,Maccone_PRL_2006}, and squeezing~\cite{noisy-unitary-estimation-1,noisy-unitary-estimation-2,Lawrie2019}. 
Quantum metrology has a wide range of applications in several fields: from detecting gravitational waves~\cite{Aasi2013,Schnabel2010,Danilishin2012,RevModPhys.86.121,PhysRevLett.123.231108, PhysRevLett.123.231107}, to optical resolution~\cite{PhysRevLett.117.190802,PhysRevX.6.031033,PhysRevLett.117.190801}, from quantum clocks~\cite{PhysRevLett.117.143004,RevModPhys.83.331,RevModPhys.87.637,Nichol2022,Nichol2022} to quantum imaging~\cite{Dowling:15,PhysRevLett.102.253601,Lugiato2002,PhysRevLett.109.123601,Genovese2016,Moreau2019}, and even in biology~\cite{TAYLOR20161,Mauranyapin2017}.

Noise is ubiquitous in 
quantum systems, and metrological tasks must account for noise in at least two distinct ways. First, noise can arise during probe preparation, the encoding process, or the measurement of the encoded state. In this case, the goal is to analyze how noise impacts the precision of a metrological task in its presence. Second, one may be interested in estimating the parameters characterizing a noisy quantum channel itself. This is particularly valuable, as knowing how much noise has affected a system allows one to better handle other quantum information tasks subjected to the same type of noise.
The first scenario has been extensively studied in various works~\cite{cirac_plenio_PhysRevLett.79.3865, PhysRevLett.112.120405,Yousefjani2017, peng2023dissipativequantumfisherinformation,PhysRevA.109.052626}, while the second, focusing on estimating the noise parameters of quantum channels, has received comparatively less attention~\cite{PhysRevA.63.042304, Akio_Fujiwara_2003, MasahitoHayashi_2010,demkowicz_2012,rafal_2014, PhysRevLett.118.100502, RAZAVIAN2019825, PRXQuantum.2.010343,rafal_2023}.

Engineering entanglement is  costly: both with respect to its preparation and its manipulation.
This naturally leads to a key question: is entanglement truly necessary to attain optimal precision in quantum metrological tasks? More specifically, does the optimal probe state - one that maximizes the quantum Fisher information - always need to be entangled?
While it is well-established that entangled probes are essential for achieving precision beyond the standard quantum limit in unitary phase estimation, the necessity of entanglement in the context of noise-parameter estimation of quantum channels remains less well-understood.

Our work aims to address this open question by investigating whether, and if yes, how much entanglement is required to attain the best metrological precision in estimation of noise parameter of channels.
We focus on a class of encoding processes that we refer to as vector encoding. We derive a necessary criterion to find optimal probe(s) for parameter estimation of these encoding processes, by using a concept that we refer to as continuous commutativity. Using this criterion, we explore the entanglement properties of the optimal probes for two representative subclasses of vector encoding: arbitrary-dimensional local depolarizing channels and local bit-flip channels.

In the estimation of the noise parameter of arbitrary-dimensional depolarizing channels using two-qudit probes, we obtain that maximally entangled two-qudit probes do not always provide optimal precision. As the depolarizing noise strength decreases, the entanglement of the optimal probe increases  \textit{step-by-step}, one Schmidt rank at a time, culminating in an abrupt transition - several Schmidt ranks at a single go - to a maximally entangled state at a critical noise strength that depends on the local Hilbert space dimension. It is worth noting that for any arbitrary-dimensional local depolarizing channel, there always exists a dimension-dependent region of noise strength in which any pure product probe serves as the optimal probe for estimating the noise strength. 
For multiparty probing, we show that for sufficiently high local depolarization noise, estimating the noise strength is optimal for fully product probes. For a low number of qubits used as probes, we numerically find that fully product states are optimal for estimation of even moderate noise levels. 
For two-qubit local bit-flip channels, we again find that entanglement is not a necessary resource: a~\textit{specific product probe} suffices to achieve the optimal precision in estimating the noise parameter, whatever the value of bit-flip noise strength is. However, unlike the case of local depolarization, even entangled states can do the same job in the whole noise range.

The remaining part of the paper is organized as follows. In Section~\ref{sec2}, we discuss some preliminary concepts of quantum parameter estimation and introduce a large class of parameter encoding, which we refer to as vector encoding, and which includes two paradigmatic quantum channels as examples: depolarizing and bit-flip channels. Our main results are presented in Sections~\ref{sec3},~\ref{sec4}, and~\ref{sec5}. In Section~\ref{sec3}, we provide a necessary criterion to identify the optimal probe for parameter estimation of any vector encoding. We term the criterion as continuous commutativity.  Section~\ref{sec4} focuses on the role of entanglement in achieving optimal precision for noise-parameter estimation of two-qudit and multi-qubit local depolarizing channels.
In Section~\ref{sec5}, we explore the necessity of entanglement for estimating the noise parameter of a two-qubit local bit-flip channel with optimal precision. Finally, in Section~\ref{sec6}, we summarize our findings.

\section{Essential background}
\label{sec2}
In this paper, we study whether entanglement is necessary to achieve the best possible precision in estimating the noise strength of arbitrary-dimensional local depolarizing channels and two-qubit local bit-flip channels.
So we will provide a brief review of single parameter quantum estimation theory, followed by discussions on the arbitrary-dimensional depolarizing channel and the single-qubit bit-flip channel in this section. Throughout the paper, we denote the $d$-dimensional state space as $\mathbb{C}^d$. 
We denote the identity operator on that space as $\mathbb{I}_{d}:\mathbb{C}^d \to \mathbb{C}^d$. 
$\{\sigma^{x},\sigma^{y},\sigma^{z}\}$ denote the three Pauli matrices. 
$\ket{0}$ and $\ket{1}$ represent the eigenvectors of $\sigma^z$ with eigenvalues $+1$ and $-1$, respectively. The $\ket{+}$ denotes eigenvector of $\sigma^x$ with eigenvalue $+1$ and $\ket{-}$ denotes eigenvector of $\sigma^x$ with eigenvalue $-1$. Also, we use $H(X) \coloneqq -X\log_2(X)-(1-X)\log_2(1-X)$ to denote the binary entropy function of $X.$

\subsection{Quantum parameter estimation}
\label{subsec: quantum_parameter_estimation}
Here, we discuss exclusively single parameter quantum estimation theory, though it can be extended to the multi-parameter case~\cite{Liu_2020,bhattacharyya2024precision,pal2025role}. We adopt a frequentist approach for parameter estimation, utilizing the local unbiasedness conditions for estimators.

Let us consider a quantum process characterized by a parameter $p$ to be estimated, with $p_0$ denoting its true value.
Let a quantum state $\rho_{0} :\mathbb{C}^{d} \to \mathbb{C}^{d}$ act as a quantum probe, and suppose that the 
process transforms the probe $\rho_{0}$ into the resulting state, $\rho_{p}$.
To estimate the parameter of interest, a quantum measurement, say $\{\Pi_z\}$, satisfying $\Pi_z \geq 0$ and $\sum_z \Pi_z = \mathbb{I}_d$, is performed on the encoded state $\rho_{p}$. According to Born's rule, the probability of obtaining the measurement outcome $z$ when performing measurement $\{\Pi_z\}$ on the state $\rho_{p}$ is given by  
$q_{p}(z) \coloneqq \Tr[\rho_{p} \Pi_z]$.

An estimator, denoted by $\hat{\Theta}(z)$ and modeled as a random variable corresponding to each measurement outcome $z$, is defined within the framework of parameter estimation to quantitatively characterize the estimation process. 
In quantum parameter estimation theory, the estimation process is quantified by the mean square error (MSE) of the estimator, given by  
\begin{align*}
    \text{MSE}(\rho_{p_0}, \Pi_z, \hat{\Theta}) &\coloneqq \int  q_{p_0}(z) (\hat{\Theta}(z) - p_0)^2 dz.
\end{align*}
In the frequentist approach, only unbiased estimators are considered to estimate the parameter which is built by imposing two constraints on estimators, known as local unbiasedness conditions. 
These are given by  
\begin{align*}
    &\int dz \, q_{p_0}(z) 
    \hat{\Theta}(z) = p_0,\quad
    &\int dz \, \hat{\Theta}(z) \frac{dq_{p} (z)}{dp} \bigg|_{p = p_0} = 1.
\end{align*}
Under these local unbiased conditions, the mean square error of estimator $\text{MSE}(\rho_{p_0}, \Pi_z, \hat{\Theta})$ coincides with the variance of the estimator, denoted as $\Delta^2 (\rho_{p_0}, \Pi_z, \hat{\Theta})$.

Note that, in order to achieve the best possible precision in estimating the parameter of interest, one needs to minimize the variance of the estimator over the choice of estimator, measurement, as well as the input state.
The classical Cram\'er-Rao bound establishes a lower bound on the variance of estimator upon minimizing over all the estimators for a fixed measurement and fixed probe. 
The lower bound on the variance of estimator is given by the inverse of $N$ times the classical Fisher information, $\mathbb{F}_\text{C}$, expressed as  
\begin{align}
 \Delta^2 (\rho_{p_0}, \Pi_z, \hat{\Theta}) &\geq \frac{1}{N\mathbb{F}_\text{C}\left(\Pi_z,\rho_{p_0}\right)},
\end{align}
where $N$ represents the number of repetitions of the experiment, and the classical Fisher information $\mathbb{F}_\text{C}$ is defined as  
\begin{equation*}
    \mathbb{F}_\text{C} \left(\Pi_z, \rho_{p_0}\right) \coloneqq \int dz \, q_{p_0}(z)
\left( \left. \frac{\partial \ln q_p(z)}{\partial p} \right|_{p = p_0} \right)^2
\end{equation*}
Moreover, by optimizing over all possible quantum measurements for a given probe, a lower bound on the classical Cram\'er-Rao bound is derived which is 
Symmetric Logarithmic Derivative Cram\'er-Rao Bound (SLD-CRB). 
The SLD-CRB is given by the inverse of $N$ times the symmetric logarithmic derivative quantum Fisher information (SLD-QFI), denoted as $\mathbb{F}_\text{Q}\left(\rho_{p_{0}}\right)$ i.e.,
\begin{align}\label{eq: QCRB}
   \Delta^2 (\rho_{p_0}, \Pi_z, \hat{\Theta}) &\geq \frac{1}{N\mathbb{F}_\text{C}\left(\Pi_z,\rho_{p_{0}}\right)}\geq \frac{1}{N\mathbb{F}_\text{Q}\left(\rho_{p_{0}}\right)},
\end{align}
where SLD-QFI is defined as $\mathbb{F}_\text{Q}\left(\rho_{p_{0}}\right) \coloneqq \Tr(\rho_{p_0} L_{p_0}^2)$ which
satisfy the relation $\mathbb{F}_\text{C} \leq \mathbb{F}_\text{Q}$.
Here, $L_{p}$ denotes the symmetric logarithmic derivative (SLD) operator, satisfying the equation  
\begin{align*}
    \frac{\partial \rho_{p}}{\partial p}\big|_{p = p_0} = \frac{1}{2} (L_p \rho_p + \rho_p L_p)\big|_{p = p_0},
\end{align*}
where $\frac{\partial \rho_{p}}{\partial p}$ represents the first-order partial derivative of the state $\rho_{p}$ with respect to $p$.  
There always exist at least a measurement that saturates the QCRB in single parameter estimation.

However, SLD-CRB lower bound does not account for optimization over the input state itself, i.e.,
the SLD-CRB lower bound of 
$\Delta^2 (\rho_{p_0}, \Pi_z, \hat{\Theta})$ is dependent on the given probe $\rho_0$, which is within our control. So, to get best achievable precision in any estimation task, one needs to optimize $\mathbb{F}_\text{Q}$ over input states. From now on we will use the term ``QFI" for the abbreviation of symmetric logarithmic derivative quantum Fisher information.

We now outline several important properties of the QFI~\cite{Liu_2020}, which will be crucial for our analysis:

\begin{itemize}
    \item \textit{Relation to fidelity:} 
    The fidelity~\cite{Jozsa01121994} between two states, say $\rho_{1}$ and $\rho_{2}$, is defined as follows: $\mathcal{F}\left(\rho_{1},\rho_{2}\right)\coloneqq\Tr\left[\sqrt{\rho^{1/2}_{1}\rho_{2}\rho^{1/2}_{1}}\right]^{1/2}$ where $\Tr\left[A\right]$ denotes the trace of the operator $A$. The QFI of the encoded state $\rho_p$ is related to fidelity between encoded states $\rho_p$ and $\rho_{p+dp}$ as~\cite{zhou2019exactcorrespondencequantumfisher} 
\begin{equation*}
\mathcal{F}\left(\rho_{p},\rho_{p+dp}\right)=1-\frac{1}{4}\mathbb{F}_{\text{Q}}\left(\rho_{p}\right)dp^{2},
\end{equation*}
in the limit $dp \rightarrow 0$.
This relationship highlights that as the infinitesimal distinguishability between nearby states parameterized by $p$ increases, the QFI also increases.

\item \textit{Convexity:} The QFI satisfies a convexity property over convex mixtures of encoded states parameterized by $p$. Specifically, for a set of encoded states $\{\rho_p^{\mu}\}$ and corresponding $p$-independent probabilities $\{a_{\mu}\}$ with $\sum_{\mu} a_{\mu} = 1$, the QFI obeys the inequality
\begin{equation*}
\mathbb{F}_{\text{Q}}\left(\sum_{\mu} a_{\mu} \rho_p^{\mu} \right) \leq \sum_{\mu} a_{\mu}\mathbb{F}_{\text{Q}}\left(\rho_p^{\mu}\right).
\end{equation*}
This implies that probabilistic mixing of encoded states cannot increase the QFI.

\item \textit{Invariance under unitary:} 
The QFI of an encoded state is invariant under unitary transformations that are independent of the parameter of interest. Specifically, if the unitary $U$ does not depend on the parameter of interest $p$, then
\begin{equation}
\label{eq: unitary invariance of QFI}
\mathbb{F}_{\text{Q}}\left(\rho_{p}\right) = \mathbb{F}_{\text{Q}}\left(U\rho_{p}U^{\dagger}\right).
\end{equation}
\end{itemize}

\subsection{Vector encoding}
Here, we will discuss a large class of encoding by quantum maps, which we will refer to as vector encoding.

We will call an encoding as 
a quantum vector encoding process, with respect to the parameter \(p\), if, in the operator sum representation, 
its action on a quantum state $\rho_{0}:\mathbb{C}^{d} \to \mathbb{C}^{d}$ produces the final state as follows:
\begin{equation}
\label{eq: defn. of alpahcp map}
\alpha_{p}\left(\rho_{0}\right):=\sum_{i} f_{i}(p)M_{i}\rho_{0}M^{\dagger}_{i},
\end{equation}
where the functions $f_i(p)$ depend on the noise parameter $p$, while the operators $M_i$ are independent of $p$ satisfying the completeness condition $\sum_i f_i(p) M_i^{\dagger} M_i = \mathbb{I}_d$. If a quantum channel can be expressed in this form - with \(p\)-independent $M_i$ and all dependence on $p$ contained in the weights $f_i(p)$ - we refer to it as a vector encoding with respect to the parameter \(p\). The corresponding Kraus operators are $\sqrt{f_i(p)} M_i$.
Common examples of vector encoding include the encoding with bit-flip channel, phase-flip channel, dephasing, and depolarizing channels, with respect to the noise strengths of the channels.
However, not all CPTP maps fall into this category. For instance, the amplitude-damping channel cannot be expressed in the form in Eq.~\eqref{eq: defn. of alpahcp map}, and hence vector encoding cannot be done with it. This falls under a category that we may term as matrix encoding, in which case all Kraus operators cannot be expressed as a product of a noise-parameter-dependent function and a noise-parameter-independent operator. 
There would be a ``scaler encoding'', when \(f_i(p)\) in Eq.~(\ref{eq: defn. of alpahcp map}) is independent of \(i\), which however is  
not possible due to the normalization property of quantum states.

\subsubsection{Depolarizing channel} 
Since the depolarizing channel is a completely positive and trace-preserving (CPTP) map, it admits a Kraus operator representation. For a single-qubit system, the Kraus operators of the depolarizing channel are given by $\{\sqrt{1 - p}\mathbb{I}_2$, $\sqrt{p/3}\sigma^x$, $\sqrt{p/3}\sigma^y$, and $\sqrt{p/3}\sigma^z\}$, where $p$ denotes the noise strength of the channel~\cite{Nielsen_Chuang_2010}. In this case, the channel acts equally on all single-qubit states with the same purity, contracting only their Bloch vectors uniformly by a factor in magnitude.
The action of a depolarizing channel can be extended to single-qudit quantum states in arbitrary dimensions. In this case, the Kraus operators are given by $\{\sqrt{1 - p}\mathbb{I}_d, \sqrt{p/(d^2-1)} \lambda^{(1)}, \cdots,\sqrt{p/(d^2-1)} \lambda^{(d^2-1)} \}$, 
where ${ \lambda^{(i)} }$ for $i = \{1, \cdots, (d^2 - 1) \}$ denotes a fixed set of generators of the $\mathbb{SU}(d)$ group corresponding to the $d$-dimensional Hilbert space.
Accordingly, the action of the arbitrary-dimensional depolarizing channel $\Lambda_p$ on an input state $\rho_0 :\mathbb{C}^{d} \to \mathbb{C}^{d}$ is given by
\begin{align}
    \Lambda_{p} (\rho_{0}) = \left(1-\frac{4p}{3}\right)\rho_{0} + \frac{4p}{3} \operatorname{Tr}(\rho_{0}) \frac{\mathbb{I}_{d}}{d}.
\end{align}
In our convention, the noise strength $p$ of the depolarizing channel is restricted to the interval $[0, 3/4]$. When $p = 0$, the channel acts trivially, leaving the input state unchanged. On the other hand, when $p = 3/4$, the channel completely depolarizes the state, mapping it to the maximally mixed state. It is also worth noting that encoding via the depolarizing channel satisfies the properties of vector encoding.

An important property of any local depolarizing channel, $(\Lambda_{p_1} \otimes \Lambda_{p_2} \otimes \cdots \otimes \Lambda_{p_n}) : \mathcal{H}_d \rightarrow \mathcal{H}_d$, acting on a multi-partite quantum state $\rho_0 : (\mathbb{C}^d)^{\otimes n} \to (\mathbb{C}^d)^{\otimes n}$, is that it commutes with local unitary operations. I.e.,
\begin{widetext}
\begin{align}
\label{eq: effect of local unitary on depolarising channel}
\Lambda{p_1} \otimes \Lambda_{p_2} \otimes \cdots \otimes \Lambda_{p_n}\left(U_{l} \rho_0 U_{l}^{\dagger}\right) = U_{l} (\Lambda_{p_1} \otimes \Lambda_{p_2} \otimes \cdots \otimes \Lambda_{p_n} \left(\rho_0\right)) U_{l}^{\dagger},
\end{align}
\end{widetext}
where 
$U_l \coloneqq \bigotimes_{i=1}^n u_i$ denotes a local unitary operation composed of unitaries $u_i$, acting individually on each of the $n$ subsystems of the state $\rho_0$. The elements withing the sets $\{p_i\}$ and $\{u_i\}$ can be the same or different from each other.

\subsubsection{Bit-flip channel}
The action of any bit-flip channel, denoted as $\beta_{p}$, on a single qubit $\rho_{0}$ can be mathematically described as 
\begin{equation*}
 \beta_{p}\left(\rho_{0}\right)=(1-p)\rho_{0}+p\sigma^{x}\rho_{0}\sigma^{x}.  
\end{equation*}
Here $p \in [0,1]$ denotes the noise strength of bit-flip channel.
Note that the action of noise in bit-flip channel satisfies the property of vector encoding. 
Moreover, it plays a central role in the development and analysis of quantum error correction (QEC) codes~\cite{Nielsen_Chuang_2010}. It is therefore interesting to estimate the noise strength of a bit-flip channel.

\begin{figure}
   \includegraphics[scale=0.32]{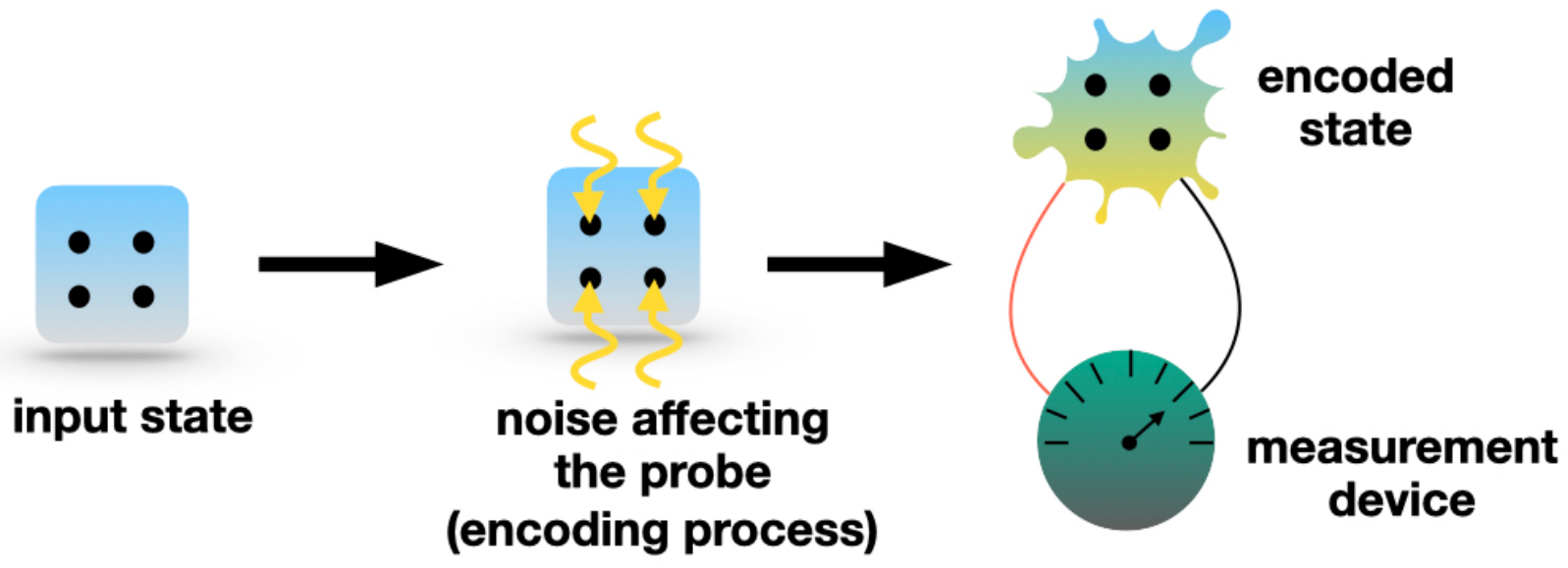}
    \caption{\textbf{Schematic diagram of setup for quantum parameter estimation.} Each quantum parameter estimation process can be divided into three steps: input state preparation, the encoding process, and finally the readout, which involves extracting statistics based on measurements performed on the encoded state. In the first step, the blue-colored box represents the input state or probe, with each dot corresponding to an individual party of the probe. In the second step, the encoding process considered in our work, the noisy channels (depicted by yellow-colored curvy lines) act locally and identically on each party of the input state. 
    The resulting encoded state is described by a blue-yellowish splatter. In the final step, measurements are performed on the encoded state and a black-lined, deep green-colored circle represents a measurement device. Based on the statistical outcomes of these measurements, the strength of the noise  channels is estimated.
    }
    \label{fig:diagram}
\end{figure}

\section{Necessary criterion for optimal probes of noise-parameter estimation in vector encoding}
\label{sec3}
The SLD-CRB for estimating any parameter of interest depends on the choice of the input quantum state. Therefore, to achieve the best possible precision in any estimation task, one must optimize the SLD-CRB over all input states.
In this work, we focus on the role of entanglement in the optimal probe for estimating a single parameter encoded via vector encoding.
Throughout the paper, we adopt a parallel estimation scheme in which the local quantum channel or vector encoding acts locally and identically on each subsystem of the multipartite probe, with the noise parameter of the local channel being the parameter of interest.
This parallel estimation scheme is described schematically in Fig.~\ref{fig:diagram}. 
In this section, we state a necessary criterion to find the optimal probe for estimating the parameter encoded with vector encoding under a parallel estimation scheme.

\begin{theorem}
\label{theorem-necessary-alpha-cp}
In a metrological task involving the estimation of any single parameter $p$ encoded through vector encoding $\alpha_p$, under a parallel estimation scheme, an input state $\rho_0: (\mathbb{C}^d)^{\otimes n} \to (\mathbb{C}^d)^{\otimes n} $ acts as an optimal probe if, in the limit $dp \to 0$, the following condition holds:
\begin{equation*}
\label{eq: first_commutation}
\left[ \rho_p, \rho_{p + dp} \right] = 0.
\end{equation*}
Here, $\rho_p \coloneqq \alpha_p^{\otimes n}(\rho_0)$ denotes the encoded state 
obtained by applying the channel $\alpha_p$ identically and locally to each of the $n$ subsystems of the input state.
\end{theorem}

\begin{remark}
The condition on the encoding and input in the theorem can be referred to as the ``continuously commutative'' property of the encoding process. Therefore the theorem states that any vector encoding process is continuously commutative for optimal probes. 
\end{remark}

\begin{proof}
Here, $\left[ A, B\right]$ denotes the commutator between matrices $A$ and $B$. Let $\alpha_{p_1}$ and $\alpha_{p_2}$ be two vector encodings characterized by the encoded parameters $p_1$ and $p_2$, respectively.
These maps act locally and identically on each subsystem of an input state $\rho_0 : (\mathbb{C}^d)^{\otimes n} \to (\mathbb{C}^d)^{\otimes n}$, resulting in the encoded states
\begin{equation*}
    \rho_{p_1} \coloneqq \alpha_{p_1}^{\otimes n}(\rho_0), \quad \rho_{p_2} \coloneqq \alpha_{p_2}^{\otimes n}(\rho_0).
\end{equation*}
Using the form of the vector encoding, given in Eq.~\eqref{eq: defn. of alpahcp map}, the encoded states $\rho_{p_1}$ and $\rho_{p_2}$ can be written as
\begin{align*}
    \rho_{p_1} &= \sum_{i} f_i(p_1)\, M_i \rho_0 M_i^{\dagger}, \quad \rho_{p_2} = \sum_{i} f_i(p_2)\, M_i \rho_0 M_i^{\dagger}.
\end{align*}

A lower bound on the fidelity between the states $\rho_{p_1}$ and $\rho_{p_2}$, denoted by $\mathcal{F}(\rho_{p_1}, \rho_{p_2})$, can be obtained as follows:
\begin{align}
    \mathcal{F}(\rho_{p_1}, \rho_{p_2}) 
    &\geq \sum_{i} \sqrt{f_i(p_1) f_i(p_2)}  \mathcal{F}\left(M_i \rho_0 M_i^{\dagger}, M_i \rho_0 M_i^{\dagger}\right) \nonumber \\
    &= \sum_i \sqrt{f_i(p_1) f_i(p_2)}, \label{eq:theorem}
\end{align}
where the first inequality follows from the strong concavity property of fidelity~\cite{Nielsen_Chuang_2010}, and the last equality uses the fact that the fidelity between two identical quantum states is equal to one.
The lower bound on the fidelity between the states $\rho_{p_1}$ and $\rho_{p_2}$, as given in Eq.~\eqref{eq:theorem}, is saturated when the encoded states commute with each other, i.e., when
$\left[\rho_{p_1}, \rho_{p_2}\right] = 0$.

From the relation between the QFI of the encoded state $\rho_p$ and the fidelity between the states $\rho_p$ and $\rho_{p+dp}$ in the limit $dp \to 0$, as discussed in Section~\ref{sec2}, it follows that the QFI is maximized over the probes when the fidelity 
$\mathcal{F}\left(\alpha^{\otimes n}_p(\rho_0), \alpha^{\otimes n}_{p+dp}(\rho_0)\right)$
is minimized.
Therefore, among all possible probes, the input states $\rho_0$ that satisfy the commutation condition $\left[ \alpha^{\otimes n}p(\rho_0), \alpha^{\otimes n}{p+dp}(\rho_0) \right] = 0$ in the limit $dp \to 0$ yield higher precision in estimating the noise parameter $p$ of the vector encoding.

This completes the proof.
\end{proof}

Let us define
\begin{equation}
\label{eq-theo-coro}
Q_i \coloneqq \sum_{\substack{j \neq i \\ f_j(p) = f_i(p)}} \left( M_i \rho_0 M_i^\dagger + M_j \rho_0 M_j^\dagger \right),
\end{equation}
where $\rho_0$ is the input state on which the vector encoding process acts on and $\{\sqrt{f_i(p)} M_i\}$ denotes the Kraus operators of an vector encoding process written in Eq.~\eqref{eq: defn. of alpahcp map}. Theorem~\ref{theorem-necessary-alpha-cp} can be simplified as follows:

\begin{corollary}
A necessary criterion for any probe to be an optimal probe for estimating the noise parameter $p$ of the vector encoding process is that
\begin{equation}
    \label{eq: indiv. zero}
    \left[Q_i, Q_j\right] = 0 \quad \forall\, i \neq j.
\end{equation}

\end{corollary}

\begin{proof}
By substituting the explicit forms of $\rho_p$ and $\rho_{p+dp}$ in terms of the Kraus operators of the vector encoding process into the necessary criterion given in Theorem~\ref{theorem-necessary-alpha-cp}, the condition for optimal input state(s) simplifies to
\begin{align}
    &\sum_{i,j} f_i(p) f_j(p+dp) \left[ M_i \rho_0 M_i^\dagger , M_j \rho_0 M_j^\dagger \right] = 0 \notag \\
    \Rightarrow &\sum_{\substack{j \neq i \\ f_j(p) = f_i(p)}} f_i(p) f_j(p+dp) \left[Q_i, Q_j\right] = 0, \label{eq:corollary}
\end{align}
where $\{Q_i\}$ are defined in Eq.~\eqref{eq-theo-coro}.  
From Eq.~\eqref{eq:corollary}, it follows that the commutation relation,
$\left[Q_i, Q_j\right] = 0; \hspace{1 mm} \forall\, i \neq j$,
must hold for any probe to be optimal in this estimation task.

This completes the proof.
\end{proof}

\section{Entanglement in optimal probe for noise-parameter estimation of depolarizing channel}
\label{sec4}

In this section, we characterize the entanglement properties of the optimal probe for estimating the noise parameter of a typical vector encoding process,
the depolarizing channel, in three cases where the probes are restricted to: (i) two-qudit (ii) arbitrary multi-qubit, and (iii) arbitrary multi-qudit dimension. Throughout this work, we restrict our analysis to pure input states while optimizing the SLD-CRB in each estimation task, since the QFI satisfies convex property over the space of encoded states, as discussed in Section~\ref{sec2}.

By combining the invariance of QFI under parameter-independent unitaries, as discussed in Eq.~\eqref{eq: unitary invariance of QFI}, with the commutation relation in between local unitaries $(U_l)$ and the local depolarizing channel $(\Lambda_{p})$, as written in Eq.~\eqref{eq: effect of local unitary on depolarising channel}, we obtain
\begin{equation*}
\mathbb{F}_{\text{Q}} \left(\Lambda_{p}^{\otimes n}\left(U_{l}\rho_{0}U_{l}^{\dagger}\right)\right) =\mathbb{F}_{\text{Q}} \left(\Lambda_{p}^{\otimes n}\left(\rho_{0}\right)\right).
\end{equation*}
Here, $U_l \coloneqq \otimes_{i=1}^n u_i$ represents a unitary operation composed of of local unitaries $u_i$ which acts locally on the $i$-th subsystem of the $n$-partite state $\rho_0$, and $\Lambda_p$ denotes the local depolarizing channel.
It implies that, for estimating the noise parameter of a depolarizing channel under the parallel estimation scheme, it suffices to consider input states in their Schmidt-decomposition form. Moreover, once an optimal input state is found, all input states related to it by local unitary operations also act as optimal probes in this estimation task.

\subsection{Two-qudit probes}
\label{two-qudit}
We analyze the entanglement properties of optimal probes for estimating a two-qudit local depolarizing channel under a parallel estimation scheme, restricting the probes to two-qudit states of arbitrary but equal dimension.

Any two-qudit pure state of arbitrary but equal dimension $d$ can, without loss of generality, be expressed in the Schmidt-decomposition form as
\begin{equation*}
    \ket{\psi_0} = \sum_{i=1}^{d} \sqrt{c_i} \ket{ii},
\end{equation*}
where $\{c_i\}$ are the Schmidt coefficients and $\{\ket{i}\}$ denotes any local orthonormal basis. In this analysis, we can consider the state $\ket{\psi_0}$ as the input state.
Based on the necessary criterion discussed in Section~\ref{sec3}, the optimal probe for estimating the noise parameter of a two-qudit local depolarizing channel under the parallel estimation scheme using two-qudit probes of equal, arbitrary dimension must have the form
\begin{equation}
\label{two-qudit-probe}
    \ket{\psi_m} = \frac{1}{\sqrt{m}} \sum_{i=1}^{m \leq d} \ket{ii},
\end{equation}
with Schmidt rank $m \in \mathbb{Z} \cap 1 \leq m \leq d$ (For detailed calculation, see Appendix~\ref{ax:form_of_optimum_state}).
Let $\mathbb{F}_{\text{Q}}(p, m, d)$ denotes the QFI associated with the input state given in Eq.~\eqref{two-qudit-probe}, characterized by Schmidt rank $m$.
It takes the following form:
\begin{align*} 
\mathbb{F}_{\text{Q}}(p, m,d) &= \frac{64}{9} \Big[ \left(1 - x\right)^2 
+ \frac{1}{d^2} A_1(x, \epsilon, d) \\
&+ \left(x^2 - \frac{1}{d^2}\right) A_2(x, \epsilon, d) 
+ 2x(1 - x) A_3(x, \epsilon, d) \Big].
\end{align*}
where $A_1(x,\epsilon,d) \coloneqq \frac{\left[1-\left(1+d^{2}\right)\epsilon-\frac{\left(1-2\epsilon\right)}{x}\right]^{2}}{d^{2}\epsilon+\left(1-\epsilon\right)^{2}+\frac{2\epsilon\left(1-\epsilon\right)}{x}}$, $A_2(x,\epsilon,d) \coloneqq \frac{\left[1-\epsilon-\frac{\left(1-2\epsilon\right)}{x}\right]^{2}}{\left(1-\epsilon\right)^{2}+\frac{2\epsilon\left(1-\epsilon\right)}{x}}$, and $A_3(x,\epsilon,d) \coloneqq \frac{\left[1-\epsilon-\frac{\left(1-2\epsilon\right)}{2x}\right]^{2}}{\left(1-\epsilon\right)^{2}+\frac{\epsilon\left(1-\epsilon\right)}{x}}$
Here we have used the notations $x\coloneqq \frac{m}{d}$ and $\epsilon\coloneqq \left(1-\frac{4p}{3}\right)$. See Appendix~\ref{ax:form_of_Fisher_information} for a detailed derivation of QFI in this estimation scenario.

In order to find the optimal probe for estimating the each noise parameter value $p$ 
of a two-qudit local depolarizing channel acting on the Hilbert space of local arbitrary dimension $d$,
one must maximize $\mathbb{F}_{\text{Q}}(p, m, d)$ over the Schmidt rank $m \in \{1, 2, \ldots, d\}$.
Let us denote the Schmidt rank that maximizes $\mathbb{F}_{\text{Q}}(p, m, d)$ by $m_0$ when the value of parameter of interest is $p$.
Here, we use the von Neumann entropy of the reduced density matrix of a pure two-qudit state as the entanglement measure~\cite{Nielsen_Chuang_2010}.
The entanglement of the optimal probe, denoted as $E$, in this estimation scenario is then given by $E=\log_{2} \left(m_{0}\right)$.
We analyze the necessity of entanglement in achieving the optimal precision in this estimation scenario across different ranges of the parameter of interest $p$ below.\\

\noindent \textbf{Case I (Limit of weak depolarization and moderate dimension):}
In the limit of the small parameter of interest $p$, and for a moderate local dimension $d$ of Hilbert space, 
the QFI corresponding to the input states in Eq.~\eqref{two-qudit-probe}, $\mathbb{F}_{\text{Q}}(p,m,d)$, simplifies as follows:
\begin{align*}
&\mathbb{F}_{\text{Q}}(p,m,d) \approx \frac{1}{p}\left(\frac{8}{3}\right)\left(1-\frac{1}{x d^{2}}\right)-\frac{32}{3}(1-x)^{2} \notag \\ 
&+\frac{64}{9} \left(1-x\right)^{2}+\frac{64}{9}\left(1-\frac{1}{xd^{2}}\right)^{2} +\frac{16}{3x}(x-2)\left(x^{2}-\frac{1}{d^{2}}\right).
\end{align*}

Note that, in this special case, $\mathbb{F}_{\text{Q}}(p, m, d)$ increases monotonically with $x$ over the allowed range $\left[\frac{1}{d}, 1\right]$.
Therefore, $\mathbb{F}_{\text{Q}}(p, m, d)$ gets maximum at $x=1$ corresponding to the maximally entangled probe. Thus, in the limit of a small parameter of interest ($p$) and for a moderately dimensional ($d$) two-qudit local depolarizing channel, the optimal probe is given by $\ket{\psi}_{\text{opt}}^{p \to 0} = \frac{1}{\sqrt{d}} \sum_{i=1}^d \ket{ii}$,
where $\{\ket{i}\}$ forms any local orthonormal basis in the Hilbert space of arbitrary dimension $d$. It corresponds to a two-qudit maximally entangled probe, up to local unitaries.\\

\noindent\textbf{Case II (Limit of strong depolarization and moderate dimension):} 
When the local depolarizing channel is highly depolarized, i.e., the noise strength of single-qudit depolarizing channel, approaches $3/4$.
In this regime of parameter of interest $p \to \frac{3}{4}$, we have 
\begin{align*}
&\mathbb{F}_{\text{Q}}(p,m,d) \approx \frac{32}{9} \left(\frac{1}{x}-1\right) +\frac{128}{9}\epsilon \left(\frac{1-x}{2x}\right) \left(1-\frac{1}{2x} \right),
\end{align*}
where $\epsilon:= (1-\frac{4}{3}p)$. In the limit of parameter of interest $p \to \frac{3}{4}$ (i.e., $\epsilon \to 0$), $\mathbb{F}_{\text{Q}}(p,m,d)$ becomes a monotonically decreasing function of $x$ over the interval $\frac{1}{d} \leq x \leq 1$. Hence, $\mathbb{F}_{\text{Q}}(p,m,d)$ will attain its maximum at $x = \frac{1}{d}$. This implies that, in the limit of strong depolarization and moderate local dimension, the optimal probe is any pure product state, up to local unitaries.\\

\begin{figure}
\includegraphics[scale=0.26]{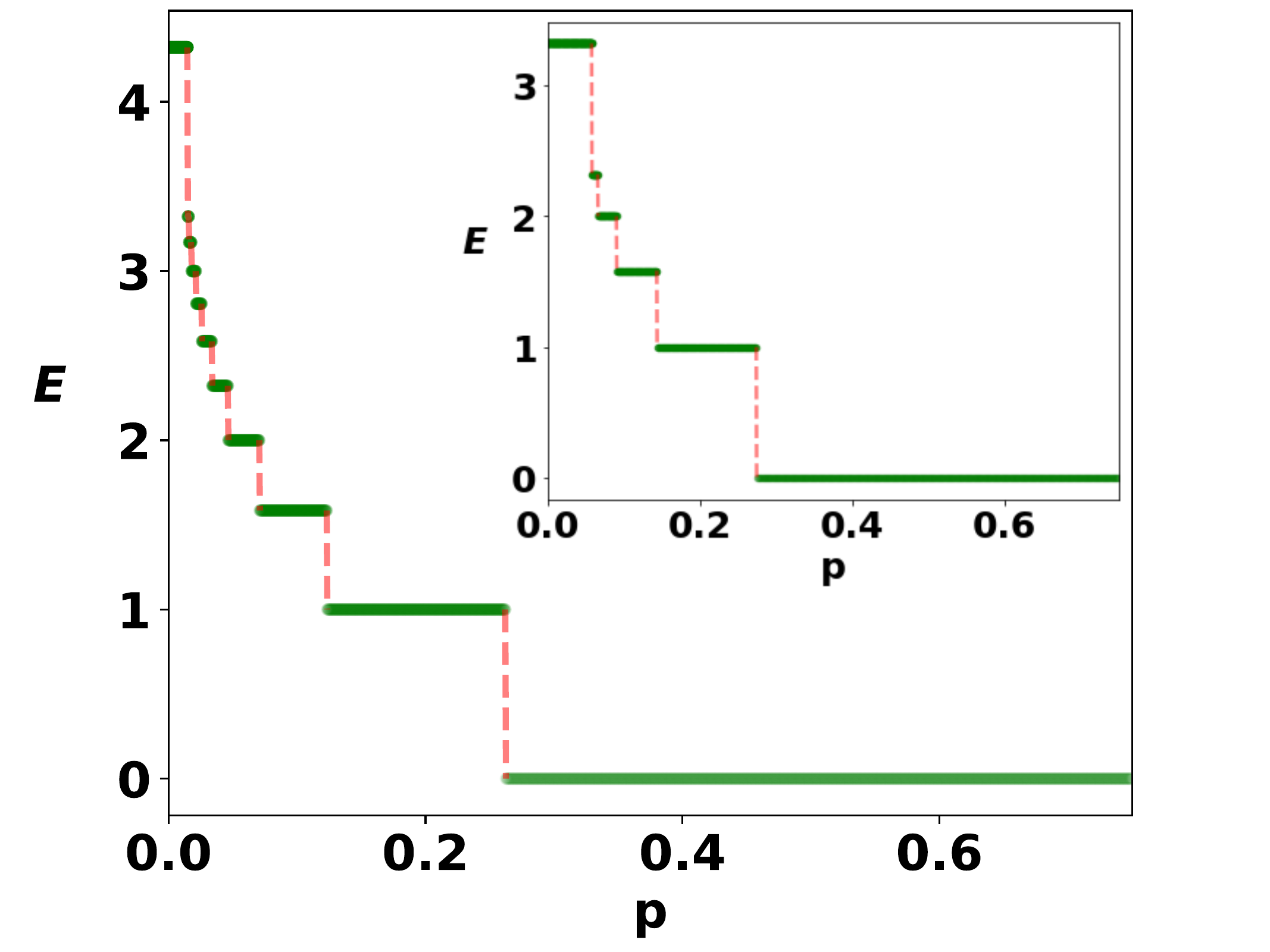}
\caption{\textbf{Behavior of optimal probe's entanglement in noise-parameter estimation of two-qudit local depolarizing channel.} We plot the entanglement of the optimal probe, denoted by $E$, as a function of the depolarizing noise strength in the estimation of the noise parameter $(p)$ of a single-qudit depolarizing channel using two-qudit probes. The main figure and its corresponding inset show results for local Hilbert space dimensions $d = 20$ and $d = 10$, respectively. The horizontal axis represents the depolarizing strength, while the vertical axis shows the entanglement of the optimal probe, quantified by the von Neumann entropy of its reduced density matrix. We observe a staircase-like structure in the entanglement of the optimal probe with respect to the depolarizing strength. The vertical axis is measured in ebits, while the horizontal axis is dimensionless.}
\label{fig:depolarising}
\end{figure}

\noindent \textbf{Case III (Moderate depolarization and large dimension):} Here, we investigate the entanglement characteristics of the optimal probe when the local depolarizing channel acts on arbitrarily large-dimensional Hilbert space, and the parameter of interest, $p$, can take any value in the range $[0,3/4]$ in this considered estimation scenario. 
In this case, we have
\begin{align} 
\label{eq: Fq derivative}
    &\frac{d\mathbb{F}_{\text{Q}}(p,xd,d)}{dx} 
  \notag \\ &= \frac{8}{p} \bigg[
        \frac{(3-4p)}{(3-4p+2px)^{2}} 
        - \frac{3}{(3-4p+4px)^{2}} \notag \\ 
    &\quad - \frac{4p^{2}}{d^{2}(3-4p)(3-4p+2px)^{2}} \notag \\ 
    &\quad + \frac{\{d^{2}(3-4p)^{2}-16p^{2}\}^{2}}
        {(3-4p)\{8pd(3-4p+2px)+xd^{3}(3-4p)^{2}\}^{2}}
    \bigg].
\end{align}
Note that both functions, $\frac{(3-4p)}{(3-4p+2px)^{2}}$ and $\frac{3}{(3-4p+4px)^{2}}$, fall monotonically at the same power with $x$ over the range $\frac{1}{d} \leq x \leq 1$, and the inequality $\frac{(3-4p)}{(3-4p+2px)^{2}} < \frac{3}{(3-4p+4px)^{2}}$ holds at both $x = 0$ and $x = 1$, and thus throughout this interval.
Furthermore, when $d$ is very large, the third term in Eq.~\eqref{eq: Fq derivative} becomes insignificant across the entire domain, and the contribution of the last term is negligible everywhere except at $x = 0$, compared to even the smallest value of the difference 
$\big(\frac{3}{(3-4p+4px)^{2}} - \frac{(3-4p)}{(3-4p+2px)^{2}} \big)$ over the domain.
Therefore, $\mathbb{F}_{\text{Q}}(p,m,d)$ decreases with increasing $m$ for all values of $p$, implying that, in this special case, pure product probes (i.e., $x = \frac{1}{d}$), up to local unitaries, serve as the optimal probes.\\

\noindent \textbf{Case IV (Moderate depolarization and arbitrary finite dimension):} Here, we discuss a more general case in the noise-parameter estimation of a two-qudit local depolarizing channel with the same local dimension: the depolarization strength lies in the intermediate regime, and the local Hilbert space dimension of probes is arbitrary but finite.

We observe a staircase-like entanglement structure in the optimal probe as a function of the parameter of interest, \( p \), as shown in Fig.~\ref{fig:depolarising}. 
As \( p \) decreases from \( \frac{3}{4} \), the Schmidt rank \( m \) of the optimal probe increases sequentially - one rank at a time - from \( m = 1 \) up to a critical value \( m = m^* \). I.e., the entanglement of the optimal probe grows \textit{step by step} as \( p \) decreases from \( \frac{3}{4} \), transitioning from product states to increasingly entangled states. At a specific, nontrivial value of \( p \), a sudden transition occurs: the optimal probe abruptly shifts from a non-maximally entangled state with \( m = m^* \) to a maximally entangled one (i.e., \( m = d \)) - jumping several Schmidt ranks at once. In other words, as \( p \) decreases from \( 3/4 \) to a critical value, the entanglement of the optimal probe increases step by step from \( m = 1 \) to \( m = m^* \); beyond this point, it abruptly jumps from \( m = m^* \) to \( m = d \). The corresponding \( m^* \) depends on the local Hilbert space dimension and is given by \( m^* = d/2 \) for even \( d \), and \( m^* = (d \pm 1)/2 \) for odd \( d \). Moreover, the value of \( p \) at which the entanglement of the optimal probe increases stepwise from \( m = 1 \) to \( m = 2 \) depends on the local Hilbert space dimension \( d \) and shifts toward zero depolarization as \( d \) increases.
This indicates that maximally entangled two-qudit probes do not always provide optimal precision; instead, there exists a dimension-dependent noise-strength regime where any pure product probe serves as the optimal probe for estimating the noise parameter of a two-qudit local depolarizing channel.

\subsection{Multi-qubit probes}
\label{multiqubit}
Here, we extend the analysis of Subsection~\ref{two-qudit} to a more general setting: the estimation of the noise parameter for a single-qubit local depolarizing channel using multi-qubit probes. All results presented here are obtained numerically.

When the probes are restricted to the Hilbert space $\mathbb{C}^2 \otimes \mathbb{C}^2 \otimes \mathbb{C}^2$ in the noise-parameter estimation of single-qubit local depolarizing channel, we find that the GHZ state ($\ket{\text{GHZ}} \coloneqq \frac{1}{\sqrt{2}} (\ket{000} + \ket{111})$), up to local unitaries, serves as the optimal probe for estimating the noise parameter $p$ in the range $0 \leq p < 0.1943$. In the range $0.1943 \lesssim p < 0.2053$, the W state ($\ket{\text{W}} \coloneqq \frac{1}{\sqrt{3}} (\ket{001} + \ket{010} + \ket{100})$), up to local unitaries, provides the optimal precision, and for $0.2053 \lesssim p < 0.3169,$ the optimum probe is bi-product states (up to local unitaries).  
Here, bi-product states refer to three-qubit states in which any two parties share a maximally entangled state, while the third party, remaining completely disentangled from the others, is in a pure single-qubit state.
In the range $0.3169 \lesssim p \leq \frac{3}{4}$, the optimal probes are fully product states up to local unitaries. This finding is consistent with earlier results obtained for two-qubit probes~\cite{PhysRevA.63.042304}, where pure product states were also found to be optimal within the same range of the noise parameter $p$. A schematic description of the entanglement structure of optimal probes across the full range of $p$ - when three-qubit states are used as probes in this estimation scenario - is provided in Fig.~\ref{fig: 3qubit_states}. 
As in the bipartite case, we can see that the entanglement of the optimal probe also exhibits a staircase-like structure across ranges of depolarization strength.
But one can quantify the amount of multipartite entanglement in many more ways compared to the two-party scenario, with examples including the geometric measure of entanglement~\cite{SHIMONY_1995,Barnum_2001,Chieh_2003}, average of bipartite negativities~\cite{HORODECKI_1996,peres_prl_1996,zyczkowski_1998,park_2000,vidal_werner_2002,plenio_2005}, relative entropy of entanglement~\cite{vedral_prl_1997,vedral_rmp_2002}, and generalized geometric measure~\cite{GGM1,GGM2}. 
Depending on the choice of the quantifier (of the multiparty entanglement in the probe), the staircase may or may not be monotonic. But regardless of the choice of the measure, after $p\approx 0.3169$, the entanglement of the optimal probe is zero across the range $0.3169 \lesssim p \leq \frac{3}{4}$, as the state of the optimal probe is fully product.

Furthermore, in the case of four-qubit probes, we find that fully product input states, up to local unitaries, continue to serve as optimal probes within the same noise interval ($0.3169 \lesssim p \leq \frac{3}{4}$), reinforcing the trend observed in both the two- and three-qubit cases.

\begin{figure} 
\hspace{-2.75 mm}
\includegraphics[scale=0.128]{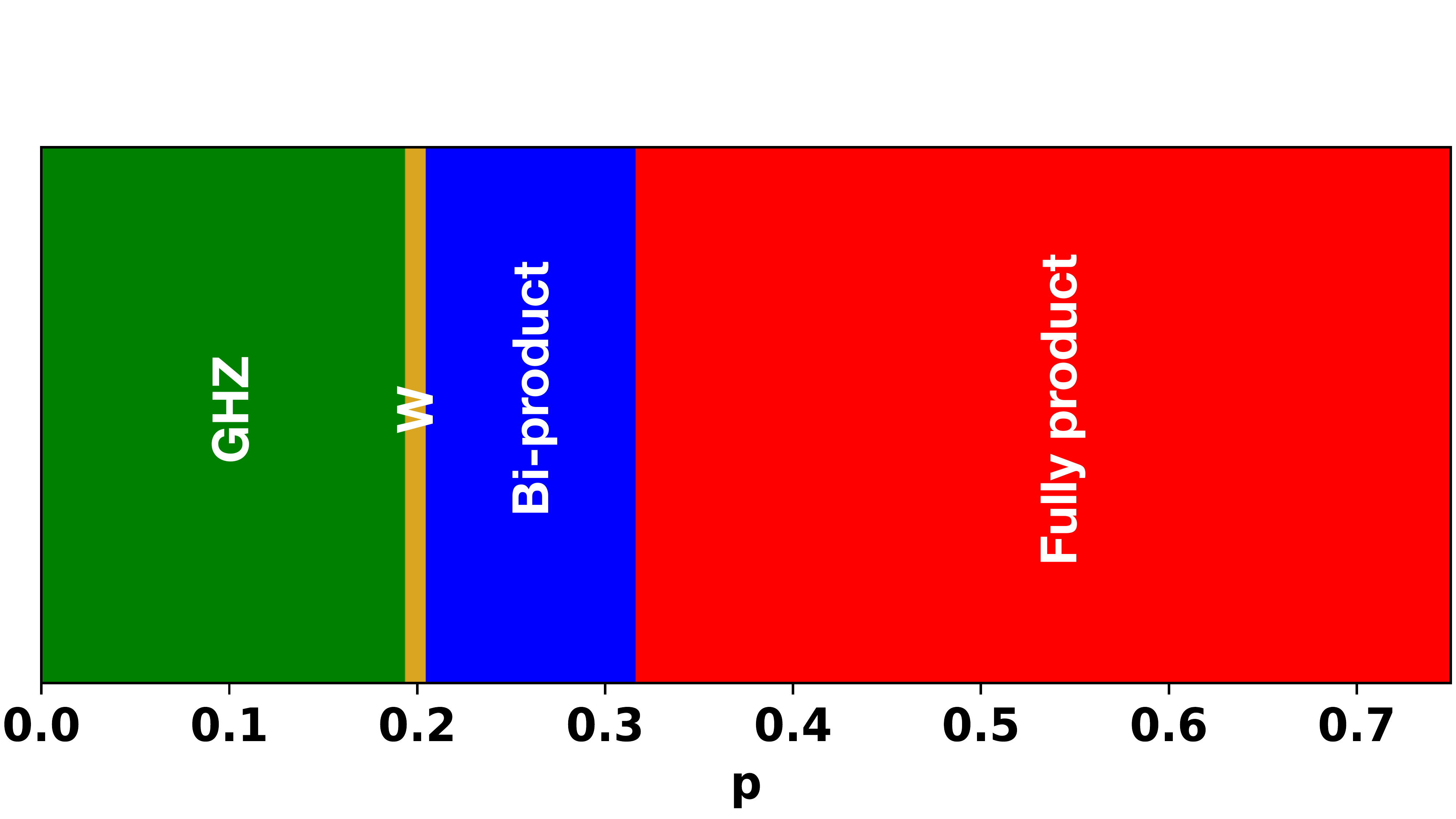
}
\caption{\textbf{Schematic presentation of optimal probe's entanglement characteristics in estimating depolarizing noise strength using three-qubit probes.} Here, we present a schematic diagram of the entanglement of the optimal probe in the case of estimating the noise strength of a local three-qubit depolarizing channel, with respect to the depolarizing noise strength.
Along the horizontal axis, the value of the dimensionless depolarizing strength, $p$, increases. Different colors represent different types of three-qubit entangled states: green, golden, blue, and red indicate the three-qubit GHZ state, three-qubit W state, three-qubit bi-product state, and fully product three-qubit state, respectively.
The GHZ state, up to local unitaries, and the W state, also up to local unitaries, serve as the optimal probe for estimating the parameter $p$ in the ranges $0 \lesssim p < 0.1943$ and $0.1943 \lesssim p < 0.2053$, respectively.
For $0.2053 \lesssim p < 0.3169$, the optimal probe is a three-qubit state in which any two parties share a maximally entangled state, while the third party - completely disentangled from the others - is in a pure single-qubit state. In the remaining range of $p$ ($0.3169 \lesssim p \leq \frac{3}{4}$), the optimal probe is any fully product probe.
    }
    \label{fig: 3qubit_states}
\end{figure}

\subsection{Multi-qudit probes}
\label{multiqudit}
In this subsection, we investigate the most general parallel noise-parameter estimation scenario for an arbitrary-dimensional single-qudit depolarizing channel, where the probes are multipartite systems.
\begin{theorem}
In the estimation of the depolarizing noise parameter of an arbitrary single-qudit depolarizing channel, within a parallel estimation scheme using multi-partite probes with an arbitrary number of parties, any product state (up to local unitaries) serves as an optimal probe for sufficiently high noise strength.
\end{theorem}

\begin{proof}
Let us consider a $n$-partite quantum state, $\rho_0 :(\mathbb{C}^{d})^{\otimes n} \to (\mathbb{C}^{d})^{\otimes n}$, where $n$, and $d$ both can be arbitrary. 
The action of local single-qudit depolarizing channel $(\Lambda_p)$ with highly depolarizing noise strength $p \to \frac{3}{4}$ and local dimension $d$ on the state $\rho_0$ can be written as
\begin{align*}
&\Lambda^{\otimes n}_{p}(\rho_{0})\\&=\left(\frac{4p}{3}\right)^{n}\left(\frac{\mathbb{I}_d}{d}\right)^{\otimes n}+\left(1-\frac{4p}{3}\right)\left(\frac{4p}{3}\right)^{n-1}\sum_{i=1}^{n}R^1_{i}\\&+ \left(1-\frac{4p}{3}\right)^2\left(\frac{4p}{3}\right)^{n-2}\sum_{i=1}^{\binom{n}{2}}R^2_{i} + \ldots + \left(1-\frac{4p}{3}\right)^n \rho_0.
\end{align*} 
Here, $R_{1}^1=\rho_{1}\otimes \left(\frac{\mathbb{I}_d}{d}\right)^{\otimes (n-1)}, \quad R_{2}^1=\frac{\mathbb{I}_d}{d}\otimes \rho_{2}\otimes \left(\frac{\mathbb{I}_d}{d}\right)^{\otimes (n-2)}, \ldots, R_n^1 = \left(\frac{\mathbb{I}_d}{d}\right)^{\otimes (n-1)}\otimes \rho_n $, and $R^2_1 = \rho_{12} \otimes \left(\frac{\mathbb{I}_d}{d}\right)^{\otimes (n-2)}, \cdots$. 
Here $\rho_{i}$ denote the state obtained
after the partial tracing out of all the parties from  $\rho_{0}$ except $i$-th party and $\{\rho_{ij}\}$ denote the states obtained
after the partial tracing out of all the parties from  $\rho_{0}$ except $i$-th and $j$-th parties with $i,j \in \{1,2, \ldots, n\}$.

Let us suppose that the state $\rho_{i}$ with $i \in \{1,2, \ldots, n\}$ can be written as $\rho_{i}=\sum_{j}p^{i}_{j}\ketbra{\psi^{i}_{j}}{\psi^{i}_{j}}$, and let us call $(1-4p/3)$ as $\epsilon$. 
In the limit $\epsilon \to 0$, which is the high noise limit, we can simplify $\Lambda^{\otimes n}_{p}$ as follows:
\begin{align*}
\Lambda^{\otimes n}_{p \to \frac{3}{4}}(\rho_{0})&\approx\left(1-n\epsilon\right)\left(\frac{\mathbb{I}_d}{d}\right)^{\otimes n}+\epsilon \sum_{i=1}^{n}R^1_{i}\\
&=\frac{\sum_{i=1}^{n}\Big[\left(1-n\epsilon\right)\left(\frac{\mathbb{I}_d}{d}\right)^{\otimes n}+n\epsilon R^1_{i}\Big]}{n}\\ 
&=\frac{\sum_{r=1}^{n}\left(\sum_{j}p^{r}_{j}\tilde{\sigma^{r}_{j}}\right)}{n}.
\end{align*}
Here, $\Tilde{\sigma^{r}_{j}}$ is $\left(\frac{\mathbb{I}_d}{d}\right)^{\otimes n}$ but the $\left(\frac{\mathbb{I}_d}{d}\right)$ of the $r$-th party is replaced by $\sigma^{r}_{j}$, where $\sigma^{r}_{j} \coloneqq (1-n\epsilon)\frac{\mathbb{I}_d}{d}+n\epsilon\ketbra{\psi^{i}_{j}}{\psi^{i}_{j}}$. 
Now, from the convexity property of QFI discussed in Section~\ref{sec2}, one can see that the maximum value of QFI occurs when, for each \( r \), \( p^{r}_{j} = 1 \) for only one value of \( j \), and zero for all others. This implies that each \( \rho_i \) must be pure to attain the maximum QFI, which happens only when the input state is fully product.

This completes the proof.
\end{proof}

\begin{figure*}
\includegraphics[scale=0.369]{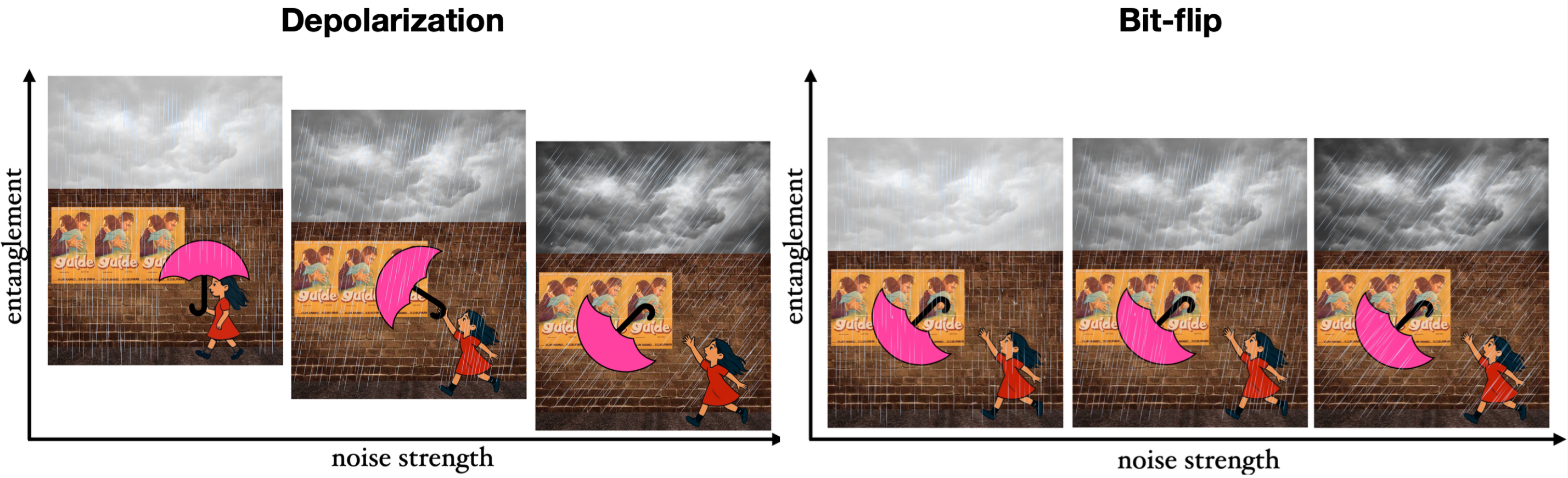}
\caption{\textbf{Depiction of results: noise estimation of depolarization vs. bit-flip channel.}
The left panel describes the necessity of entanglement for achieving optimal precision in noise-parameter estimation of a local depolarizing channel, while the right panel does the same for a local bit-flip channel. Each panel is divided into three segments, representing different regimes of noise strength, increasing from left to right. Within each segment, the entanglement properties of the optimal probe are conveyed through a real-life metaphor: a girl walking in the rain while holding an umbrella.
In this metaphor, the intensity of the thunderstorm, which grows from left to right in each panel, corresponds to increasing noise strength, and the use of the umbrella symbolizes the necessity of entanglement in the input probe to attain optimal precision.
In the depolarizing channel case (left panel), entanglement becomes progressively less relevant as noise strength increases. Accordingly, as the storm intensifies, the umbrella becomes increasingly cumbersome, reflecting that entanglement is no longer beneficial in highly noisy regimes.
In contrast, for the bit-flip channel (right panel), the girl is never shown holding an umbrella, regardless of the storm’s intensity. This signifies that entanglement is never required - whether the bit-flip noise is weak or strong, a product probe suffices for optimal estimation.
It is worth noting that, although not depicted in the presentation, the optimal precision decreases with increasing noise strength in both cases, regardless of the probe used.
}
\end{figure*}

\section{Entanglement of optimal probe in noise-parameter estimation of bit-flip channel}
\label{sec5}
This section analyzes the entanglement properties of the optimal probe in noise-parameter estimation for a two-qubit local bit-flip channel under the parallel estimation scheme.

\begin{theorem}
Let us consider the estimation of the noise parameter of a two-qubit local bit-flip channel under a parallel estimation scheme, where the probes are restricted to two-qubit states. Then, the optimal probe state is given by
\begin{equation}
\label{eq: two qb bf state}
\ket{\psi}^{\text{opt}}_{\text{bf}} (\theta) \coloneqq \frac{\ket{++}+\ket{+-}+\ket{-+}+e^{i\theta} \ket{–}}{2}; \quad \theta \in [0, 2\pi),
\end{equation}
and the maximum QFI corresponding to the optimal probes is given by $\mathbb{F}_{\text{Q}}^{\text{max}} = 2/p(1-p)$.
\end{theorem}

\begin{remark}
The entanglement of this optimal probe state, quantified by the von Neumann entropy of the reduced density matrix, is given by
\(E\left(\ket{\psi}^{\text{opt}}_{\text{bf}} (\theta)\right)= 
H(\cos^2(\theta/4))\)
which varies in the range $0 \leq E\left(\ket{\psi}^{\text{opt}}_{\text{bf}} (\theta)\right) \leq 1$.
This implies that even a product input state, $\ket{\psi}^{\text{opt}}_{\text{bf}} (\theta=0)$,
can serve as an optimal probe for this estimation task.
\end{remark}

\begin{proof}
When each qubit of any pure two-qubit state, $\rho_0 (\coloneqq \ket{\psi_0}\bra{\psi_0})$, undergoes a two-qubit local bit-flip channel with noise strength \( p \in [0,1] \), the resulting state can be written as
\begin{align*}
\beta_p^{\otimes 2}(\rho_0) = (1 - p)^2 \rho_0 + 2p(1 - p)\, \rho_{12} + p^2 \rho_3,
\end{align*}
where $\rho_1 \coloneqq \ket{\psi_1}\bra{\psi_1}, \hspace{2 mm} \rho_2 \coloneqq \ket{\psi_2}\bra{\psi_2}, \hspace{2 mm} \rho_3 \coloneqq \ket{\psi_3}\bra{\psi_3}, \hspace{2 mm} \ket{\psi_1} \coloneqq \sigma^x \otimes \mathbb{I}_2 \ket{\psi_0}, \hspace{2 mm}
\ket{\psi_2} \coloneqq \mathbb{I}_2 \otimes \sigma^x \ket{\psi_0}$,  $\ket{\psi_3} \coloneqq \sigma^x \otimes \sigma^x \ket{\psi_0}$, and $\rho_{12}=(\rho_1 + \rho_2)/2$.
Any pure two-qubit state can be expressed as
\begin{equation*}
    \ket{\psi_0} = a\ket{++} + b \ket{+-} + c\ket{-+} + d\ket{--},
\end{equation*}
with \( |a|^2 + |b|^2 + |c|^2 + |d|^2 = 1 \).  
The states $\ket{\psi_1}, \ket{\psi_2}$, and $\ket{\psi_3}$ will have the following forms in the basis of $\{\ket{++}, \ket{+-}, \ket{-+}, \ket{--}\}$:
\begin{align*}
    \ket{\psi_1} &= a\ket{++} + b \ket{+-} - c\ket{-+} - d\ket{--}, \\
    \ket{\psi_2} &= a\ket{++} - b \ket{+-} + c\ket{-+} - d\ket{--},\\
    \ket{\psi_3} &= a\ket{++} - b \ket{+-} - c\ket{-+} + d\ket{--}.
\end{align*}
According to Theorem~\ref{theorem-necessary-alpha-cp}, the states \( \rho_0 \), \( \rho_{12} \), and \( \rho_3 \) must commute with each other for $\rho_0$ to be the optimal probe.  
Moreover, the commutation relation \( [\rho_0, \rho_3] = 0 \) implies \( [\rho_1, \rho_2] = 0 \).  
This, in turn, means that the pure states \( \rho_1 \) and \( \rho_2 \) are either the same or orthogonal.  
Accordingly, we now consider two separate cases: (I) \( \rho_1 = \rho_2 \), and (II) \( \rho_1 \) and \( \rho_2 \) are orthogonal.

We begin by considering the case where \( \rho_1 = \rho_2 \).  
In this situation, Theorem~\ref{theorem-necessary-alpha-cp} reduces to the following set of commutation relations:
$[\rho_0, \rho_1] = 0, [\rho_0, \rho_3] = 0$, and $[\rho_1, \rho_3] = 0$.
These imply that the three states \( \rho_0 \), \( \rho_1 \), and \( \rho_3 \) are mutually commuting. But $\rho_1=\rho_2$ implies $\rho_0=\rho_3$. 
Under these constraints, two distinct sub-cases can arise, which we analyze in detail below:
\begin{itemize}
    \item \textit{(I-a) All four states $\rho_0$, $\rho_1$, $\rho_2$, and $\rho_3$ are identical}:  

Then, \( \rho_0 = \rho_1 = \rho_2 = \rho_3 = \ket{++}\bra{++} \), and the channel output remains unchanged:  
\( \beta_p^{\otimes 2}(\rho_0) = \ket{++}\bra{++} \).  
Hence, the QFI corresponding to the parameter of interest \( p \) is \( \mathbb{F}_{\text{Q}}^1 = 0 \).

\item \textit{(I-b) The states $\rho_0$ or $\rho_3$ is orthogonal to $\rho_1$ or $\rho_2$}:  
Here, the eigenvalues of \( \beta_p^{\otimes 2}(\rho_0) \) are \( \{(1-p)^2 + p^2, 0, 2p(1-p), 0\} \),  
and the QFI is given by  
$\mathbb{F}_{\text{Q}}^2 = \frac{4(1-2p)^2}{p^2+(1- p)^2}+\frac{2(1-2p)^2}{p(1-p)}$.

\end{itemize}

Next, we consider the second case, where \( \rho_1 \) and \( \rho_2 \) are orthogonal. But this can not be physically true when $\rho_0=\rho_3$. So, 
given the commutation relations \( [\rho_0, \rho_1] = 0 \), \( [\rho_0, \rho_3] = 0 \), and \( [\rho_1, \rho_3] = 0 \), without the loss of generality, the following sub-cases can arise in this case (with all states in below are expressed in the same basis):

\begin{itemize}

    \item \textit{(II-a) \( \rho_0 = \emph{diag}\{1, 0, 0, 0\} \), \( \rho_3 = \emph{diag}\{0, 1, 0, 0\} \), and \( \rho_{12} = \emph{diag}\{\frac{1}{2}, \frac{1}{2}, 0, 0\} \):} 
    The QFI in this case is $\mathbb{F}_{\text{Q}}^3 = \frac{1}{p(1 - p)}$.

    \item \textit{(II-b) \( \rho_0 = \emph{diag}\{1, 0, 0, 0\} \), \( \rho_3 = \emph{diag}\{0, 1, 0, 0\} \), and \( \rho_{12} = \emph{diag}\{0, \frac{1}{2}, \frac{1}{2}, 0\} \):}
    The quantum Fisher information in this case is  $\mathbb{F}_{\text{Q}}^4 = \frac{1 + p}{p(1 - p)}$.

    \item \textit{(II-c) \( \rho_0 = \emph{diag}\{1, 0, 0, 0\} \), \( \rho_3 = \emph{diag}\{0, 1, 0, 0\} \), and \( \rho_{12} = \emph{diag}\{0, 0, \frac{1}{2}, \frac{1}{2}\} \):}
    The QFI corresponding to $p$ in this sub-case is  $\mathbb{F}_{\text{Q}}^5 = \frac{2}{p(1 - p)}$.

    \item \textit{(II-d) \( \rho_0 = \emph{diag}\{1, 0, 0, 0\} \), \( \rho_3 = \emph{diag}\{0, 1, 0, 0\} \), and \( \rho_{12} = \emph{diag}\{\frac{1}{2}, 0, \frac{1}{2}, 0\} \):} 
    The QFI in this case is given by $\mathbb{F}_{\text{Q}}^6 = \frac{2 - p}{p(1 - p)}$.
\end{itemize}

Note that, among all the sub-cases discussed in both scenarios, the maximum quantum Fisher information corresponding to the noise parameter of the two-qubit local bit-flip channel in this estimation task is given by $\mathbb{F}_{\text{Q}}^{\text{max}} = \frac{2}{p(1 - p)}$ across the entire range of $p$, i.e. $p\in[0,1]$.  
This maximum QFI is achieved when the states \( \rho_0, \rho_1, \rho_2 \), and \( \rho_3 \) are mutually orthogonal.  
Exploiting the orthogonality of these states, along with the unit-trace condition, one can see that the input state attaining this maximal QFI is the one given in Eq.~\eqref{eq: two qb bf state}.
Furthermore, the eigenvalues of the reduced density matrix of the state in Eq.~\eqref{eq: two qb bf state} are given by  
$\text{diag}\{\cos^2(\theta/4),\sin^2(\theta/4)\}$ with $\theta \in [0, 2\pi).$
Thus, it completes the proof.
\end{proof}

For any $\theta \in [0, 2\pi)$, the input state given in Eq.~\eqref{eq: two qb bf state} serves as an optimal probe for this estimation task. Since the entanglement of this state can range from zero to the maximal value depending on $\theta$, it follows that entanglement is not a necessary resource: even a specific product probe can achieve the best possible precision in estimating the noise parameter of any two-qubit local bit-flip channel under the parallel estimation scheme.

\section{Conclusion}
\label{sec6}
In quantum metrology, one of the central objectives is to determine the best attainable precision in estimating a parameter that characterizes a quantum process - be it a unitary operation or a more general quantum 
evolution. Achieving this optimal precision generally requires an optimization over the input probe state, the measurement strategy, and the classical estimator.

Noise is almost universal to quantum systems, and there is a large body of important work that deals with quantum parameter estimation in presence of noise. Somewhat less explored is the area of estimation of the noise itself, and this is what we wished to address here.

Since entanglement is a costly resource, we are led to the following question: Is entanglement necessary to achieve optimal precision in such noise-parameter estimation tasks? More precisely, does the optimal probe that maximizes the quantum Fisher information necessarily require to be entangled?
In unitary-encoding phase estimation, entangled probes are known to be essential to  surpass the standard quantum limit. In contrast, for noise-parameter estimation of quantum channels, the role of entanglement remains relatively underexplored. In this work, our aim was to address this gap by  studying whether, and if yes, how much entanglement is required to attain the best metrological precision in noise-parameter estimation.

We have studied the role of entanglement in achieving optimal precision for estimating the noise parameter of a large class of quantum channels, which we referred to as vector encoding. We began by deriving a necessary criterion for identifying optimal input probe states for estimating the noise parameter in these channels.
We then investigated the entanglement characteristics of optimal probes for two representative subclasses: arbitrary-dimensional local depolarizing channels and two-qubit local bit-flip channels.

For two-qudit local depolarizing channels, we observed that maximally entangled two-qudit probes are not always optimal. Instead, as the depolarizing noise strength decreases, the entanglement of the optimal probe increases in a stepwise manner - one Schmidt rank per step, culminating in an abrupt transition - of several Schmidt ranks - to a maximally entangled state at a critical noise strength, which depends on the local Hilbert space dimension. Moreover, for any arbitrary-dimensional two-party local depolarizing channel, there always exists a range of noise strength, dependent on the local Hilbert space dimension, within which any pure product probe serves as the optimal choice for the estimation.
Furthermore, extending our analysis to multipartite probes with an arbitrary number of subsystems, we have shown that in the regime of extreme depolarization of the qudit depolarizing channel, any product input state achieves the lowest estimation error.
Also, in the noise-parameter estimation of a multi-qubit local depolarizing channel, we observe that for a certain fixed range of the noise parameter - independent of the number of parties in the probe - the optimal precision is achieved using fully product states.

For two-qubit local bit-flip channels, we have found that entanglement is not a necessary resource: a specific product state suffices to achieve the optimal precision in estimating the noise parameter. However, unlike the case of local depolarizing channels, entangled states can also provide the same optimal precision.

\section*{Acknowledgment}
We acknowledge the use of Armadillo and Nlopt for quantum information and computation. 
PG acknowledges support from the ``INFOSYS scholarship for senior students'' at Harish-Chandra Research Institute, India.

\twocolumngrid
\bibliography{met_one}
\onecolumngrid
\section*{Appendix}
\appendix
In Appendix~\ref{ax:form_of_optimum_state}, we show the process of finding the form of optimal two-qudit input states, using the necessary criterion of optimal states for vector encoding discussed in Sec.~\ref{sec3}. Next, considering the form of this two-qudit optimal probe, we derive the expression of the corresponding QFI in Appendix~\ref{ax:form_of_Fisher_information} .

\section{Optimal two-qudit probes for local depolarizing channels}
\label{ax:form_of_optimum_state}
Due to Schmidt decomposition, any bipartite pure state $\rho_0 :(\mathbb{C}^{d})^{\otimes 2} \to (\mathbb{C}^{d})^{\otimes 2}$ can be written as $\rho_0 = \ket{\psi_0} \bra{\psi_0}$, with $\ket{\psi_0}=\sum _{i=1}^{d}\sqrt{c_{i}}\ket{ii}$ and $\sum_{i=1}^d c_i = 1$. Here $\{\ket{i}\}$, for $i=\{1,2, \ldots d\}$, denotes any local orthonormal basis in the Hilbert space $\mathbb{C}^{d}$.
The action of two-qudit local depolarizing channel, $\Lambda_p$, on the state $\rho_0$ is given by
\begin{align}\label{eq: depol_encoded}
    \rho_{p} &\coloneqq \Lambda_p^{\otimes 2}(\rho_{0}) \notag \\
    &=\left(1-\frac{4p}{3}\right)^2 \rho_{0}+ \left(\frac{4p}{3}\right)^{2}\left(\frac{\mathbb{I}_d}{d}\otimes\frac{\mathbb{I}_d}{d}\right) + \left(\frac{4p}{3}\right)\left(1-\frac{4p}{3}\right) 
    \left(\frac{\mathbb{I}_d}{d}\otimes \rho_1
    + \rho_2 \otimes \frac{\mathbb{I}_d}{d}\right)  \\
     &= \left(1-\frac{4p}{3}\right)^2 \sum_{i,j=1}^{d}\sqrt{c_{i}c_{j}}\ket{ii}\bra{jj} + \left(\frac{4p}{3}\right)^{2}\frac{1}{d^{2}}\sum_{i,j=1}^{d}\ket{ij}\bra{ij} + \left(\frac{4p}{3}\right)\left(1-\frac{4p}{3}\right) \frac{1}{d}\sum_{i,j=1}^{d}(c_{i}+c_{j})\ket{ij}\bra{ij} \notag.
\end{align}
According to Theorem~\ref{theorem-necessary-alpha-cp}, we must have $\left[\rho_0, \rho_p\right] = 0$. This implies that
\begin{align}
\label{eq: com1}
&\left[\sum_{i,j=1}^{d}\sqrt{c_{i}c_{j}}\ket{ii}\bra{jj},\sum_{i,j=1}^{d}(c_{i}+c_{j})\ket{ij}\bra{ij}\right]=0 \notag \\
&\implies \sum_{i,j=1}^{d}\sqrt{c_{i}c_{j}}(c_{i}-c_{j})\ket{ii}\bra{jj}=0.
\end{align}
unless $p$ is exactly $\frac{3}{4}$, and for $p=3/4$ the commutation relation is automatically satisfied.
The Eq.~\eqref{eq: com1} implies that $c_{i}=0$ for some $i$  and all other $c_{i}$ are equal. Thus, we can write the form of the optimal input probe as
\begin{equation}\label{eq: optimal_input_depol}
    \ket{\psi _m}=\frac{1}{\sqrt{m}}\sum _{i=1}^{m\leq d}\ket{ii}.
\end{equation}

\section{Optimal quantum Fisher information
for local two-qudit depolarizing channels
}
\label{ax:form_of_Fisher_information}
We now substitute the input state of the form given in Eq.~\eqref{eq: optimal_input_depol} into the expression for the encoded state in Eq.~\eqref{eq: depol_encoded}, resulting in the encoded state denoted by \( \rho_p(m,d) \). It is important to note that \( \rho_p(m,d) \) automatically satisfies the commutation relation $[\rho_p, \frac{d\rho_p}{dp}] = 0$,
which ensures that both operators are diagonalizable in the same basis. Therefore, the QFI in this case is given by
$\sum_i \frac{1}{p_i} \left( \frac{dp_i}{dp} \right)^2$,
where \( \{p_i\} \) are the eigenvalues of \( \rho_p(m,d) \).
The complete set of \( d^2 \) eigenvalues and their corresponding degeneracies are listed in Table~\ref{table: eigenvalues}. 
Using these eigenvalues, we can compute the QFI as follows:
\begin{table}
\begin{tabular}{|c|c|}
\hline
Eigenvalues & Degeneracy \\
\hline
$\left(\frac{4p}{3d}\right)^{2}+\frac{8p}{3dm}\left(1-\frac{4p}{3}\right)+\left(1-\frac{4p}{3}\right)^{2}$ & 1 \\
\hline
$\left(\frac{4p}{3d}\right)^{2}+\frac{8p}{3dm}\left(1-\frac{4p}{3}\right)$ & $(m^{2}-1)$ \\
\hline
 $\left(\frac{4p}{3d}\right)^{2}+\frac{4p}{3dm}\left(1-\frac{4p}{3}\right)$ & $2m(d-m)$ \\
\hline
$\left(\frac{4p}{3d}\right)^{2}$ & $(d-m)^{2}$ \\
\hline
\end{tabular}
\caption{\textbf{The list of eigenvalues and corresponding degeneracies of the encoded state $\rho_p(m,d)$}. The entries in the first column of the table represent the eigenvalues of the state, while the second column indicates the degeneracies corresponding to those eigenvalues.}
\label{table: eigenvalues}
\end{table}

\begin{align*}
    \mathbb{F}_{\text{Q}}&=
    \frac{\left[\ \frac{32p}{9d^{2}}+\frac{8}{3dm}\left(1-\frac{8p}{3}\right)-\frac{8}{3}\left(1-\frac{4p}{3}\right) \right]^{2}}{\frac{16p^{2}}{9d^{2}}+\frac{8p}{3dm}\left(1-\frac{4p}{3}\right)+\left(1-\frac{4p}{3}\right)^{2}} \notag \quad +(m^{2}-1)\frac{\left[\ \frac{32p}{9d^{2}}+\frac{8}{3dm}\left(1-\frac{8p}{3}\right)\right]^{2}}{\frac{16p^{2}}{9d^{2}}+\frac{8p}{3dm}\left(1-\frac{4p}{3}\right)} \notag \quad +2m(d-m)\frac{\left[\ \frac{32p}{9d^{2}}+\frac{4}{3dm}\left(1-\frac{8p}{3}\right)\right]^{2}}{\frac{16p^{2}}{9d^{2}}+\frac{4p}{3dm}\left(1-\frac{4p}{3}\right)} \notag \\
    &+ (d-m)^{2}\frac{64}{9d^2}.
\end{align*}

\end{document}